\documentclass[svgnames,dvipsnames,11pt]{article}

\usepackage{xspace}
\usepackage{comment}
\usepackage{todonotes}
\usepackage{xparse}
\usepackage{mathtools}
\usepackage{pdfpages}

\usepackage[linesnumbered,ruled,vlined]{algorithm2e}

\usepackage{amsmath,amstext,amssymb,amsthm}
\usepackage[bookmarks,colorlinks,breaklinks]{hyperref}
\usepackage{boxedminipage}

\hypersetup{linkcolor=blue,citecolor=blue,filecolor=blue,urlcolor=blue}
\newcommand{\lref}[2][]{\hyperref[#2]{#1~\ref*{#2}}}
\renewcommand{\eqref}[2][]{\hyperref[#2]{(\ref*{#2})}}

\newtheorem{definition}{Definition}

\newtheorem{lemma}{Lemma}
\newtheorem{theorem}{Theorem}

\newtheorem{observation}{Observation}

\newtheorem{Reduction Rule}{Reduction Rule}

\newtheorem{step}{Step}

\usepackage{vmargin}
\setmarginsrb{1in}{1in}{1in}{1in}{0pt}{0pt}{0pt}{7mm}

\newcommand{\mc}{\textsc{Mixed Cut}}
\newcommand{\mmcufull}{\textsc{Mixed Multiway Cut-Uncut}}
\newcommand{\bmmcufull}{\textsc{Border-Mixed Multiway Cut-Uncut}}

\newcommand{\mmcu}{\textsc{MMCU}}
\newcommand{\bmmcu}{\textsc{B-MMCU}}
\newcommand{\bpvcfull}{{\sc Bipartite Partial Vertex Cover}}
\newcommand{\bpvc}{\textsc{BPVC}}
\newcommand{\R}{\mathcal{R}}
\newcommand{\I}{\mathcal{I}}
\newcommand{\X}{\mathcal{X}}
\newcommand{\Z}{\mathcal{Z}}
\newcommand{\CS}{\mathcal{L}}
\newcommand{\PP}{\mathcal{P}}
\newcommand{\F}{\mathcal{F}}
\newcommand{\V}{V^\infty}
\newcommand{\sol}{{\sf sol}}
\newcommand{\bigs}{{\sf big}}

\newcommand{\No}{{\sc No}}
\newcommand{\Yes}{{\sc Yes}}
\newcommand{\Oh}{{\mathcal{O}}}

\newcommand{\name}[1]{\textsc{#1}}
\newcommand{\parnamedefn}[4]{
\begin{center}
\begin{boxedminipage}{1\textwidth}
\begin{tabbing}
\name{#1} \\
\emph{Input:} \hspace{1cm} \= \parbox[t]{12cm}{#2} \\
\emph{Parameters:}            \> \parbox[t]{12cm}{#3} \\
\emph{Question:}             \> \parbox[t]{12cm}{#4} \\
\end{tabbing}
\vspace{-0.5cm}
\end{boxedminipage}
\end{center}
}

\newcommand{\probnamedefn}[3]{
\begin{center}
\begin{boxedminipage}{1\textwidth}
\begin{tabbing}
\name{#1} \\
\emph{Input:} \hspace{1cm} \= \parbox[t]{12cm}{#2} \\
\emph{Output:}             \> \parbox[t]{12cm}{#3} \\
\end{tabbing}
\vspace{-0.5cm}
\end{boxedminipage}
\end{center}
}

\usepackage{titlesec}

\title{A Parameterized Algorithm for {\sc Mixed-Cut}}

\author{Ashutosh Rai\thanks{Institute of Mathematical Sciences, Chennai, India. Emails: 
\texttt{\{ashutosh|saket\}@imsc.res.in}}  \and M. S. Ramanujan\thanks{University of Bergen, Norway. Email: \texttt{M.Ramanujan@uib.no}}
\and Saket Saurabh\addtocounter{footnote}{-2}\footnotemark
}

\begin{document}
\maketitle

\thispagestyle{empty}

%

%
\begin{abstract}
The classical Menger's theorem states that in any undirected (or directed) graph $G$, given a pair of vertices $s$ and $t$, the maximum number of vertex (edge) disjoint paths is equal to the minimum number of vertices (edges) needed to disconnect from $s$ and $t$. This min-max result can be turned into a polynomial time algorithm to find the maximum number of vertex (edge) disjoint paths  as well as the minimum number of vertices (edges) needed to disconnect $s$ from $t$. In this paper we study a mixed version of this problem, called  {\sc Mixed-Cut}, where we are given an undirected graph $G$, vertices $s$ and $t$, positive integers $k$ and $l$ and the objective is to test whether there exist 
a  $k$ sized vertex set $S\subseteq V(G)$ and an $l$ sized edge set $F\subseteq E(G)$ such that deletion of $S$ 
and $F$ from $G$ disconnects from $s$ and $t$. We start with a small observation that this problem is NP-complete 
and then study this problem, in fact a much stronger generalization of this, in the realm of parameterized complexity. In particular we study the \mmcufull\ problem where along with a set of terminals $T$, we are also given an equivalence relation $\R$ on $T$, and the question is whether we can delete at most $k$ vertices and at most $l$ edges such that connectivity of the terminals in the resulting graph respects $\R$. Our main results is a fixed parameter algorithm for \mmcufull\ using the method of recursive understanding introduced by Chitnis et al. (FOCS 2012).
\end{abstract}


\newpage 
\setcounter{page}{1}

\section{Introduction}

Given a graph, a typical {\it cut problem} asks for finding a set of vertices or edges such that their removal from the graph makes the graph satisfy some separation property. The most fundamental version of the cut problems is \textsc{Minimum Cut}, where given a graph and two vertices, called {\it terminals}, we are asked to find the minimum sized subset of vertices (or edges) of the graph such that deleting them separates the terminals. 
Minimum cut is can be applied to several problems, which at first do not look like cut problems. Examples include well studied problems such as \textsc{Feedback Vertex Set}~\cite{ChenLLOR08} and \textsc{Odd Cycle Transversal}~\cite{ReedSV04}.
Even though \textsc{Minimum Cut} is known to be in polynomial time for both edge and vertex versions, it becomes NP-hard for many of its generalizations. 

Two of the most studied generalizations of \textsc{Minimum Cut} problem which are NP-hard are \textsc{Multiway Cut} and \textsc{Multicut}. In the \textsc{Multiway Cut} problem, we are given a set of terminals, and we are asked to delete minimum number of vertices (or edges) to separate the terminals from each other. This problem is known to be NP-hard when the number of terminals is at least three. In the \textsc{Multicut} problem, given pairs of terminals, we are asked to delete minimum number of vertices (or edges) so that it separates all the given terminal pairs. The \textsc{Multicut} problem is known to be NP-hard when the number of pairs of terminals is at least three. 

These problems are one of the most well studied problems in all those algorithmic paradigms that are meant for coping with NP-hardness, such as approximation algorithms~\cite{AAC07,ENSS98,GVY96,GVY04,KKSTY04,RaviS02}. The aim of this paper is to look at a generalization of \textsc{Minimum Cut} problem in the realm of parameterized complexity. The field of parameterized complexity tries to provide efficient algorithms for  NP-complete problems by going from the classical view of single-variate measure of the running time to a multi-variate one. It aims at getting algorithms of running time $f(k)n^{\Oh(1)}$, where $k$ is an integer measuring some aspect of the problem. These algorithms are called fixed parameter tractable (FPT) algorithms and the integer $k$  is called the {\em parameter}. In most of the cases, the solution size is taken to be the parameter, which means that this approach gives faster algorithms when the solution is of small size. For more background on parameterized complexity, the reader is referred 
to the monographs \cite{ParameterizedComplexityBook,FG06,RN}.

Cut problems were looked at under the realm of Parameterized Complexity by Marx~\cite{Marx06} for the first time, who showed that \textsc{Multiway Cut} is FPT when parameterized by the solution size and \textsc{Multicut} is FPT when parameterized by the solution size plus the number of terminals. Subsequently, a lot of work has been done on cut problems in the field of parameterized complexity~\cite{BDT11,CLL09,CHM13,KT11,KPPW12,MOR13,MR14}. Recently, Chitnis et al.~\cite{CCHPP12} introduced the technique of {\it randomized contractions} and used that to solve the \textsc{Unique Label Cover} problem. They also show that the same techniques can be applied to solve a generalization of \textsc{Multiway Cut} problem, namely \textsc{Multiway Cut-Uncut}, where an equivalence relation $\R$ is also supplied  along with the set of terminals and we are to delete minimum number of vertices (or edges) such that the terminals lie in the same connected of the resulting graph if and only if they lie in the same equivalence class of $\R$.

Another interesting generalization of the \textsc{Minimum Cut} problem is the one where we allow both vertices and edges to be deleted to achieve the desired separation properties. In an instance of \textsc{Mixed Cut} problem, we are asked whether we can separate the two given terminals by deleting at most $k$ vertices and at most $l$ edges. The problem can also be looked at as asking to delete $k$ vertices so that there are at most $l$ edge disjoint paths between the two terminals.

 Even though the vertex and edge versions of \textsc{Minimum Cut} problem are polynomial time solvable, we show that allowing deletion of both, the vertices and the edges, makes the \textsc{Mixed Cut} problem NP-hard. To show that, we use a simple reduction from the \bpvcfull\  problem which was recently shown to be NP-hard~\cite{AS14,JV12}. Hence, \mc\  becomes interesting from the parameterized complexity point of view. 

The \textsc{Mixed Cut} problem can also be generalized in the same way as \textsc{Minimum Cut} in some cases. In the \textsc{Mixed Multiway Cut} problem, which is an analogue of \textsc{Multiway Cut}, given a set of terminals, we are asked to delete at most $k$ vertices and $l$ edges to separate the terminals from each other. Similarly, we can define \textsc{Mixed Multiway Cut-Uncut (MMCU)} as an analogue of \textsc{Multiway Cut-Uncut} problem where we have to preserve the connectivity between the terminals which belong to the same equivalence class, while separating the terminals which belong to different equivalence classes, by deleting at most $k$ vertices and at most $l$ edges. We defer the formal definition of the problem to the next section. It is easy to see that \mmcu\ not only generalized \textsc{Mixed Cut} and \textsc{Mixed Multiway Cut}, but also both edge and vertex versions of \textsc{Multiway Cut} and \textsc{Multiway Cut-Uncut} problems. 

Studying the parameterized complexity of \mmcu\ is the main focus of this paper. We show that \mmcu\ is FPT when parameterized by the solution size, i.e. $k+l$. By doing so, we show that \textsc{Mixed Cut} is also FPT when parameterized by the solution size.

We use the irrelevant vertex technique introduced by Chitnis et al.~\cite{CCHPP12} to solve the problem. The main observation is that if there is a small vertex separation which divides the graph into big parts, then we can recursively reduce the size of one of the big parts. Otherwise, the graph is highly connected, and the structure of the graph can be exploited to obtain a solution. The techniques and the presentation of the paper are heavily borrowed from that of~\cite{CCHPP12}, but we do need many technical modifications to make the algorithm work for \mmcu. Allowing edge deletions makes us modify many of the initial operations on the graph. The algorithm in~\cite{CCHPP12} returns the minimum solutions in the recursive steps. Since we allow both edge and vertex deletion, there is no clear ordering on the solutions, and hence we need to look for solutions of all possible sizes while making the recursive call.

\section{Preliminaries}\label{Preliminaries}

In this section, we first give the notations and definitions which are used in the paper. Then we state some basic properties of mixed-cuts and some known results which will be used later in the paper. 

\noindent
{\bf Notations and Definitions:} For a graph $G$, we denote the set of vertices of the graph by $V(G)$ and the set of edges of the graph by $E(G)$. We denote $|V(G)|$ and $|E(G)|$ by $n$ and $m$ respectively, where the graph is clear from context.
For a set $S\subseteq V(G)$, the {\it subgraph of $G$ induced by $S$} is denoted by $G[S]$ and it is defined as the subgraph of $G$ with vertex set $S$ and edge set $\{(u,v) \in E(G) :u,v\in S\}$ and the subgraph obtained after deleting $S$ is denoted as $G- S$. For $F \subseteq E(G)$, by $V(F)$ we denote the set $\{v \ | \ \exists u \text{ such that } uv \in F\}$.  For a set $Z=V' \cup E'$ where $V' \subseteq V(G)$ and $E' \subseteq E(G)$, by $G(Z)$ we denote the subgraph $G'=(V' \cup V(E'),E')$. For a tuple $\X = (X,F)$ such that $X \subseteq V(G)$ and $F \subseteq E(G)$, by $G-\X$ we denote the graph $G' = (V(G) \setminus X, E(G) \setminus F)$ and by $V(\X)$ we denote the vertex set $X \cup V(G(E))$. All vertices adjacent to a vertex $v$ are called neighbours of $v$ and the set of all such vertices is called {\it open} neighbourhood of $v$, denoted by $N_G(v)$. For a set of vertices $S \subseteq V(G)$, we define $N_G(S) = (\cup_{v \in S} N(v)) \setminus S$. We drop the subscript $G$ when the graph is 
clear from the context. 

We define the \textsc{Mixed Cut} and \mmcufull\ problems as follows.  

\parnamedefn{\mc}{A multigraph $G$, vertices $s,t \in V(G)$, integers $k$ and $l$.}{$k,l$}{Does there exist $X \subseteq V(G)$ and $F \subseteq E(G)$ such that $|X|\leq k$, $|F| \leq l$ and $s$ and $t$ are in different connected components of $G-(X,F)$?}

\parnamedefn{\mmcufull\ (\mmcu)}{A multigraph $G$, a set of terminals $T \subseteq V(G)$, and equivalence relation $\R$ on the set $T$ and integers $k$ and $l$.}{$k,l$}{Does there exist $X \subseteq (V(G) \setminus T)$ and $F \subseteq E(G)$ such that $|X|\leq k$, $|F| \leq l$ and for all $u,v \in T$, $u$ and $v$ belong to the same connected component of $G-(X,F)$ if and only if $(u,v) \in \R$?}

%

We say that a tuple $\X = (X,F)$, where $X \subseteq V(G) \setminus T$ and $F \subseteq E(G)$, is a solution to a \mmcu\ instance $\I = (G, T, \R, k, l)$ if $|X|\leq k$, $|F| \leq l$ and for all $u,v \in T$, $u$ and $v$ belong to the same connected component of $G-(X,F)$ if and only if $(u,v) \in \R$. We define a partial order on the solutions of the instance $\I$. For two solutions $\X=(X,F)$ and $\X'=(X',F')$ of a \mmcu\ instance $\I$, we say that $\X' \leq \X$ if $X' \subseteq X$ and $F' \subseteq F$. We say that a solution $\X$ to an \mmcu\ instance $\I$ is {\it minimal} if there does not exist another solution $\X'$ to $\I$ such that $\X' \neq \X$ and $X' \leq \X$. For a solution $\X=(X,F)$ of an \mmcu\ instance $\I = (G, T, \R, k, l)$ and $v \subseteq V(G)$, we say that $\X$ {\it affects} $v$ if either $v \in X$ or there exists $u \in V(G)$ such that $uv \in F$.
\begin{observation}\label{obs}
 If $\X = (X,F)$ is a minimal solution to a \mmcu\ instance $\I = (G, T, \R, k, l)$, then none of the edges in $F$ are incident to $X$.
\end{observation}

\begin{lemma}[\cite{CCHPP12}]\label{lemma:setfamily}
Given a set $U$ of size $n$ together with integers $0 \leq a, b \leq n$, one can in $\Oh(2^{\Oh(\min(a,b) \log(a+b))} n \log n)$ time construct a family $\F$ of at most $\Oh(2^{\Oh(\min(a,b) \log(a+b))} \log n)$ subsets of $U$, such that the following holds: for any sets $A, B \subseteq U$, $A \cap B = \emptyset$, $|A| \leq a$, $|B| \leq b$, there exists a set $S \in \F$ with $A \subseteq S$ and $B \cap S = \emptyset$.
\end{lemma}

\begin{definition}[\cite{CCHPP12}]
Let $G$ be a connected graph and $V^{\infty} \subseteq V (G)$ a set of undeletable vertices. A triple $(Z, V_1, V_2)$ of subsets of $V(G)$ is called a $(q,k)$-good node separation, if $|Z| \leq k$, $Z \cap V^{\infty} = \emptyset$, $V_1$ and $V_2$ are vertex sets of two different connected components of $G-Z$ and $|V1 \setminus V^{\infty}|, |V_2 \setminus V^{\infty}| >q$.
\end{definition}

\begin{definition}[\cite{CCHPP12}]
Let $G$ be a connected graph, $V^{\infty} \subseteq V (G)$ a set of undeletable vertices, and $T_b \subseteq V(G)$ a set of border terminals in $G$. A pair $(Z, (V_i)_{i=1}^l)$ is called a $(q,k)$-flower separation
in $G$ (with regard to border terminals $T_b$), if the following holds:
\begin{itemize}
 \item $1 \leq |Z| \leq k$ and $Z \cap V^{\infty} = \emptyset$; the set $Z$ is the {\it core} of the flower separation $(Z, (V_i)_{i=1}^l)$;
 \item $V_i$ are vertex sets of pairwise different connected components of $G-Z$, each set $V_i$ is a {\it petal} of the flower separation $(Z, (V_i)_{i=1}^l)$;
 \item $V (G) \setminus (Z \cup \bigcup_{i=1}^l V_i)$, called a {\it stalk}, contains more than $q$ vertices of $V \setminus V^{\infty}$;
 \item for each petal $V_i$ we have $V_i \cap T_b = \emptyset$, $|Vi \setminus V^{\infty}| \leq q$ and $N_G(V_i) = Z$;
 \item $|(\bigcup_{i=1}^l V_i) \setminus V^{\infty}| >q$.
\end{itemize}
\end{definition}

\begin{lemma}[\cite{CCHPP12}]\label{lemma:nodesep}
Given a connected graph $G$ with undeletable vertices $V^{\infty} \subseteq V (G)$ and integers $q$ and $k$, one may find in $\Oh (2^{\Oh(\min(q,k) \log(q+k))} n^3 \log n)$ time a $(q, k)$-good node separation of G,
or correctly conclude that no such separation exists.
\end{lemma}

\begin{lemma}[\cite{CCHPP12}]\label{lemma:flowersep}
Given a connected graph $G$ with undeletable vertices $V^{\infty} \subseteq V (G)$ and border terminals $T_b \subseteq V (G)$ and integers $q$ and $k$, one may find in $\Oh (2^{\Oh(\min(q,k) \log(q+k))} n^3 \log n)$  time a $(q, k)$-flower separation in $G$ w.r.t. $T_b$, or correctly conclude that no such flower separation exists.
\end{lemma}

\begin{lemma}[\cite{CCHPP12}]\label{lemma:highconnected}
If a connected graph $G$ with undeletable vertices $V^{\infty} \subseteq V (G)$ and border terminals $T_b \subseteq V (G)$ does not contain a $(q, k)$-good node separation or a $(q, k)$-flower separation w.r.t. $T_b$ then, for any $Z \subseteq V(G) \setminus V^{\infty}$ of size at most $k$, the graph $G - Z$ contains at most $(2q + 2)(2^k -1) + |T_b | + 1$ connected components containing a vertex of $V(G) \setminus V^{\infty}$, out of which at most one has more than $q$ vertices not in $V^{\infty}$.

\end{lemma}


\section{NP-completeness of \mc}

We prove that \mc\ in NP-complete by giving a reduction from the \bpvcfull\ problem which is defines as follows. 

\probnamedefn{\bpvcfull\ (\bpvc)}{A bipartite graph $G=(X \uplus Y, E)$, integers $p$ and $q$}{Does there exist $S \subseteq V(G)$ such that $|S| \leq p$ and at least $q$ edges in $E$ are incident on $X$?}

\begin{theorem}[\cite{AS14,JV12}]
 \bpvc\ is NP-complete. 
\end{theorem}
For an instance of \bpvc, we assume that the given bipartite graph does not have any isolated vertices, as a reduction rule can be applied in polynomial time which takes care of isolated vertices and produces an equivalent instance. Given an instance $(G,p,q)$ of \bpvc\ where $G=(X \uplus Y, E)$ is a bipartite graph, we get an instance $(G',s,t,k,l)$ of \mc\ as follows. To get the graph $G'$, we introduce two new vertices $s$ and $t$ and add all edges from $s$ to $X$ and $t$ to $Y$. More formally, $G'=(V',E')$ where $V' = V(G)\cup \{s,t\}$ and $E'=E \cup \{sx\ |\ x \in X\} \cup \{ty \ | \ y \in Y\}$. Then we put $k=p$ and $l=m-q$, where $m= |E|$. It is easy to see that $(G,p,q)$ is a \Yes\ instance of \bpvc\ if and only if $(G',s,t,k,l)$ is a \Yes\ instance of \mc, and hence we get the following theorem.

\begin{theorem}
 \mc\ is NP-complete even on bipartite graphs. 
\end{theorem}

 \section{An algorithm for \mmcu}

In this section, we describe the FPT algorithm for \mmcu. To that end, we first describe how we can bound the number of equivalence classes in a connected component of the given graph.

\subsection{Reduction of Equivalence Classes}

We first show that if there exists a vertex which has a large number of vertex disjoint paths to terminals which are in different equivalence classes, then that vertex has to be part of the solution.

\begin{lemma}\label{lemma:disjointpaths}
 Let $\I = (G, T, \R, k, l)$ be a \mmcu\ instance and let $v \in V (G) \setminus T$. Assume that there exist $k + l + 2$ paths $P_1, P_2 , \ldots , P_{k+l+2}$ in $G$, such that:
\begin{itemize}
 \item for each $1 \leq i \leq k + l + 2$, the path $P_i$ is a simple path that starts at $v$ and ends at $v_i \in T$;
 \item any two paths $P_i$ and $P_j$, $i \neq j$, are vertex disjoint except for the vertex $v$. 
\item  for any $i \neq j$, $(v_i , v_j) \notin \R$.
\end{itemize}
Then for any solution $(X,F)$ of $\I$ we have $v \in X$.
\end{lemma}

\begin{proof}
For a contradiction assume that $v\notin X$. Since for $i \neq j$, any two paths $P_i$ and $P_j$ are vertex disjoint except for the vertex $v$, we have that after we delete vertices of $X$ and edges of $F$ there are still two paths remaining among the given collection, say $P_a$ and $P_b$. But then we get a path from some $v_a\in T$ to some vertex $v_b \in T$ in $G - (X,F)$ by concatenating $P_a$ and $P_b$ such that $(v_a,v_b) \notin \R$. This contradicts the fact that $(X,F)$ is a solution to $\I$. 
\end{proof}

\begin{lemma}\label{lemma:disjointpathalgo}
 Let $\I = (G, T, \R, k, l)$ be a \mmcu\ instance. For any $v \in V(G)$, we can verify of $v$ satisfies conditions of Lemma~\ref{lemma:disjointpaths} in $\Oh((k+l)n^2)$ time. 
\end{lemma}

\begin{proof}
To verify whether a vertex  $v\in V(G)$ satisfies the conditions of Lemma~\ref{lemma:disjointpaths} we do as follows. Given a graph $G$ we transform it into a directed network $D$ as follows. For every edge $uv\in E(G)$ we have two directed arcs $(u,v)$  and $(v,u)$ in $E(D)$ and for every equivalence class $Z\in \R$ we make a new vertex $t_z$ and give a directed arc from $z\in Z$ to $t_z$ (that is we add $(z,t_z)$. Finally, we add  a super-sink $t$ and give directed arc $(t_z,t)$ for every  $Z\in \R$. Furthermore, give every arc unit capacity. Now to check whether $v$ satisfies the conditions of Lemma~\ref{lemma:disjointpaths} all we need to check whether there is a flow of size $k+l+2$ in $D$ from $v$ to $t$. This can be done by running $k+l+2$ rounds of augmenting path and thus the running time is upper bounded by $\Oh((k+l+2)\cdot (|E(D)|+|V(D)|))$. Thus the claimed running time follows. 
\end{proof}

\begin{step}\label{step:deletevertex}
 For each $v \in V (G)$, if $v$ satisfies the conditions of Lemma~\ref{lemma:disjointpaths}, delete $v$ from the graph and decrease $k$ by one; if $k$ becomes negative by this operation, return \No. Then restart the algorithm.
\end{step}

Using the algortihm described in Lemma~\ref{lemma:disjointpathalgo}, each application of Step~\ref{step:deletevertex} takes $\Oh((k+l)n^3)$ time. As Step~\ref{step:deletevertex} can not be applied more than $k$ times, all applications of the step take $\Oh(k(k+l)n^3)$ time. Now we show how application of Step~\ref{step:deletevertex} gives a bound on number of equivalence classes.

\begin{lemma}\label{lemma:equivalenceclassbound}
 Let $\I = (G, T, \R, k, l)$ be the instance of \mmcu~which we get after the exhaustive application of Step~\ref{step:deletevertex}. If there exists a connected component in $G$ containing terminals from more than $(k+l)(k+l+1)$ equivalence classes of $\R$, then $\I$ is a \No\ instance of \mmcu.
\end{lemma}

\begin{proof}
 For a contradiction assume that $\I$ is a \Yes\ instance of \mmcu. We will show that if this is the case then there exists an opportunity to apply Step~\ref{step:deletevertex}. Let $(X,F)$ be a solution to $\I$ and $X'$ be a set of vertices containing all of $X$ and exactly one vertex from each of the edges in $F$. Clearly, the size of $X'$ is at most $k+l$. Let $C_1, \ldots, C_s$ be  the connected components of $G-X'$ that contains at least one terminal vertices. Since  $G$ contains terminals from more than $(k+l)(k+l+1)$ equivalence classes of $\R$ and each connected component of $G-X'$ contains vertices from at most one equivalence class we have that $s>(k+l)(k+l+1)$. Thus by Pigeon Hole Principle  there exists a vertex $v\in X'$ that has neighbors in at least $k+l+2$ connected components. These connected components together with $v$ imply that we could apply Step~\ref{step:deletevertex}. This contradicts our assumption that on $\I$ Step~\ref{step:deletevertex} is no longer applicable. Thus indeed $\
I$ 
is a \No\ instance of \mmcu.
\end{proof}

\begin{step}\label{step:equivalenceclass}
If there exist $u, v \in T$ such that $u$ and $v$ lie in different connected components of $G$ or there exists a connected component of $G$ with terminals of more than $(k+l)(k+l+1)$
equivalence classes or $\R$, return \No.
\end{step}

Correctness of Step~\ref{step:equivalenceclass} follows from Lemma~\ref{lemma:equivalenceclassbound} and it can be applied in $\Oh(n^2)$ time. In the next step, take care of disconnected graphs which helps us assume that the input graph is connected. 


\begin{step}\label{step:disconnected}
Let $C_1, C_2, \ldots C_t$ be the vertices corresponding to the connected components of the graph $G$ in the input instance $\I = (G, T, R, k, l)$. For each $C_i$ and for every $k' \in \{0,1,\ldots,k\}$ and $l' \in \{0,1,\ldots,l\}$ we pass the instance $\I_{k',l'}^i = (G,T \cap C_i, \R|_{T \cup C_i}, k', l')$ to the next step. If for every $i \in \{1,\ldots,t\}$ there exists $\I_{k_i',l_i'}^i$  such that the subroutine returns \Yes\ for $\I_{k_i',l_i'}^i$ and $\sum_{i \in \{1,\ldots,t\}} k_i' \leq k$, $\sum_{i \in \{1,\ldots,t\}} l_i' \leq \l$, then return \Yes, otherwise return \No.  
\end{step}

The correctness of this step is easy to see, as after the application of Step~\ref{step:equivalenceclass}, each equivalence class is confined to at most one connected component of the graph. The step can be applied in $\Oh(n^2)$ time and gives rise to $\Oh(kln)$ subproblems. After application of Step~\ref{step:disconnected}, we can assume that $G$ is connected and that the number of equivalence classes in $G$ are bounded by $(k+l)(k+l+1)$. 

\subsection{Operations on the Graph}

\begin{definition}
 Let $\I = (G, T, \R, k,l)$ be an \mmcu~instance and let $v \in V (G) \setminus T$. By bypassing a vertex $v$ we mean the following operation: we delete the vertex $v$ from the graph
and, for any $u_1,  u_2 \in N_G (v)$, we add an edge $(u_1,u_2)$ if it is not already present in $G$.
\end{definition}

\begin{definition}
Let $\I = (G, T, \R, k,l)$ be an \mmcu~instance and let $u,v \in  T$. By identifying  vertices $u$ and $v$, we mean the following operation: we make a new set $T' = (T \setminus \{u,v\}) \cup \{x_{uv}\}$ and for each edge of the form $uw \in E(G)$ or $vw \in E(G)$, we add an edge $x_{uv}$ to $E(G)$. Observe that the operation might add parallel edges.  
\end{definition}

\begin{lemma}\label{lemma:bypass}
 Let $\I = (G, T, \R, k,l)$ be a \mmcu\ instance, let $v \in V (G) \setminus T$ and let $\I' = (G', T, \R, k,l)$ be the instance $\I$ with $v$ bypassed. Then:
\begin{itemize}
 \item if $\X=(X,F)$ is a solution to $\I'$, then $\X$ is a solution to $\I$ as well;
 \item if $\X=(X,F)$ is a solution to $\I$ and $v \notin X$ and for all $u \in N(v)$ $vu \notin F$ then $\X$ is a solution to $\I'$ as well.
\end{itemize}
\end{lemma}

\begin{proof}
We first prove the first item in the lemma. Suppose  $\X=(X,F)$ is a solution to $\I'$ but it is not a solution to $\I$. This implies that either there exists $(x,y)\in \R$ and $x$ and $y$ belong to two different connected components of $G-\X$ or $(x,y)\notin \R$ and they are in the same connected component of $G-\X$.  In the first case we know that $x$ and belongs to the same connected  component of  $G'-\X$ and thus we know that the two connected components of $G-\X$ 
which contain $x$ and $y$ has neighbors of $v$. This contradicts that  $x$ and $y$ belong to two different connected components of $G-\X$. For the later case if there is a path containing $x$ and $y$ then it must contain $v$ but then in $G'$ we have an edge between the end-points of $v$ on this path, which in turn would imply a path between $x$ and $y$ in $G'-\X$. One can prove the second statement along the similar lines. 
\end{proof}

\begin{lemma}\label{lemma:certainedge}
 Let $\I = (G, T, \R, k,l)$ be a \mmcu\ instance and let $u, v \in T$ be two different terminals with $(u, v) \notin \R$, such that $uv \in E(G)$, then for any solution $\X = (X,F)$ of $\I$, we have $uv \in F$.
\end{lemma}

 The proof of the Lemma~\ref{lemma:certainedge} follows from the fact that any solution must delete the edge $uv$ to disconnect $u$ from $v$. The proof of the next lemma follows by simple observation that $u$ and $v$ have at least $k+l+1$ internally 
vertex disjoint paths or from the fact that $(u, v) \in \R$ and thus after deleting the solution they must belong to the same connected component and thus every minimal solution does not use the edge $uv \in E(G)$.  

\begin{lemma}\label{lemma:identify}
 Let $\I = (G, T, \R, k,l)$ be a \mmcu\ instance and let $u, v \in T$ be two different terminals with $(u, v) \in \R$, such that $uv \in E(G)$ or $|N_G(u) \cap N_G(v)| > k+l$. Let $\I'$ be instance $\I$ with terminals $u$ and $v$ identified. Then set of minimal solutions of $\I$ and $\I'$ is the same.  
\end{lemma}

%

\begin{lemma}\label{lemma:deleteterminals}
Let $\I = (G, T, \R, k,l)$ be an \mmcu\ instance and let $V' = \{v_1, v_2, \ldots, v_t\} \subseteq T$ be different terminals of the same equivalence class of $\R$, pairwise nonadjacent and such that
$N_G(u_1) = N_G(u_2) = \cdots =N_G(u_t) \subseteq V (G) \setminus T$  and $t >l+2$. Let $\I'$ be obtained from $\I$ by deleting all but $l+2$ terminals in $U$ (and all pairs that contain the deleted terminals in $\R$). Then the set of minimal solutions to $\I$ and $\I'$ are equal.
\end{lemma}

\begin{proof}
Without loss of generality, let the set deleted terminals be $U' := \{u_{l+3}, \ldots , u_t\}$ and let $N:= N_G(u_1) = N_G(u_2) = \cdots =N_G(u_t)$. Let $u,v \in V(G) \setminus U'$. We claim that for any minimal solution $\X = (X,F)$ of $\I$, $u$ and $v$ are in the same connected component of $G-\X$ if and only if they are in the same connected component of $G'-\X$. The backward direction is trivial. To show the forward direction, we just need to show that $F$ does not contain any edges incident on $U$, because then if a path visits a vertex in $U'$, then we can redirect it via one of $u_i$'s which remain in the graph. For $v \in X \cap N$, by Observation~\ref{obs}, we have that $F$ does not contain any edges incident on $v$. For $v \in N \setminus X$, there are $u_a, u_b \in U \setminus U'$ which are adjacent to $v$ in $G-\X$. This gives us that $v$ lies in the same connected component of $G-\X$ as $U$, and hence a minimal solution does not contain any of the edges between $x$ and $U$. This shows that $\X$ 
is a minimal solution for $\I'$ as well. 

For the converse, let $\X$ be a minimal solution to $\I'$. We just need to show that if $u, v \in T$ such that $v \in T \setminus U'$ and $u \in U'$, they lie in same connected component of $G-\X$ if and only if $(u,v) \in \R$. Since $l \geq 0$, we have that $|U \setminus U'| \geq 2$.
 Now it is easy to see that $N \subsetneq X$ because otherwise the vertices of $U \setminus U'$ will not be in the same connected component of  $G' - \X$.  Hence, there exists $x \in N \setminus X$. Since $|N_{G'}(X) \cap U'| = l+2$, $v$ is adjacent to at least $2$ vertices of $U'$ in $G' - \X$. Without loss of generality, let $u_1, u_2 \in U \setminus U'$ be these two vertices. Now let $v \in T \setminus U'$ and  $u \in U'$. If $(v, u_1) \in \R$, then $v$ and $x$ belong to the same connected component of $G'-\X$ and hence $v$ and $u$ belong to same connected component of $G-\X$. Similarly if $(v, u_1) \notin \R)$, then $v$ and $x$ belong to the different connected components of $G'-\X$ and hence $v$ and $u$ belong to different connected components of $G-\X$. This concludes the proof of the lemma.
\end{proof}

\begin{lemma}\label{lemma:deleteparallel}
 Let $\I = (G, T, \R, k,l)$ be an \mmcu\ instance and let $uv \in E(G)$ be an edge with multiplicity more than $l+1$. Then for any minimal solution $\X=(X,F)$ of $\I$, $F$ does not contain any copies of $uv$.  
\end{lemma}

\begin{proof}
 If $\{u, v\} \cap X \neq \emptyset$, then by Observation~\ref{obs}, we have that none of the copies of $uv$ are in $F$. Otherwise, $F$ contains at most $l$ copies of edge $uv$. Let $\X' = (X, F \setminus \{uv\})$. Then we have that for any two $x,y \in V(G)$, $x$ and $y$ are adjacent in $G-\X$ if and only if they are adjacent in $G-\X'$, contradicting the minimality of $\X$.  
\end{proof}

a minimal solution for $\I'$ as well. 
%
%
%

\subsection{Borders and Recursive Understanding}
In this section, we define the bordered problem and describe the recursive phase of the algorithm. Let $\I = (G, T, \R, k,l)$ be an \mmcu\ instance and let $T_b \subseteq V(G) \setminus T$ be a set of border terminals, where $|T_b| \leq 2(k+l)$. Define $\I_b = (G, T, \R, k,l, T_b)$ to be an instance of the bordered problem. By $\mathbb{P}(\I_b)$ we define the set of all tuples $\PP=(X_b, E_b, \R_b, k',l')$, such that $X_b \subseteq T_b$, $E_b$ is an equivalence relation on $T_b \setminus X_b$, $\R_b$ is an equivalence relation on $T \cup (T_b \setminus X_b)$ such that $E_b \subseteq \R_b$ and $\R_b|_T = \R$, $k'\leq k$ and $l' \leq l$. For a tuple  $\PP=(X_b, E_b, \R_b, k',l')$, by $G_{\PP}$ we denote the graph $G \cup E_b$, that is the graph $G$ with additional edges $E_b$. 

We say that as tuple $\X=(X,F)$ is a solution to $(\I_b, \PP)$ where $\PP=(X_b, E_b, \R_b, k',l')$ if $|X| \leq k'$, $|F| \leq l'$ and for all $u,v \in T \cup (T_b \setminus X_b)$, $u$ and $v$ belong to the same connected component of $G_{\PP}-(X,F)$ if and only if $(u,v) \in \R_b$. We also say that $\X$ is a solution to $\I_b = (G, T, \R, k,l, T_b)$ whenever $\X$ is a solution to $\I = (G, T, \R, k,l)$. Now we define the bordered problem as follows. 

\probnamedefn{\bmmcufull(\bmmcu)}{An \mmcu\ instance $\I = (G, T, \R, k,l)$ with $G$ being connected and a set $T_b \subseteq V(G) \setminus T$ such that $|T_b| \leq 2(k+l)$; denote $\I_b = (G, T, \R, k,l, T_b).$}{For each $\PP=(X_b, E_b, \R_b, k',l') \in \mathbb{P}(\I_b)$, output a $\sol_\PP = \X_\PP$ being a minimal solution to $(\I_b, \PP)$, or $\sol_\PP =\bot$ if no solution exists.}

It is easy to see that \mmcu\ reduces to \bmmcu, by putting $T_b = \emptyset$. Also, in this case, any answer to \bmmcu\ for $\PP=(\emptyset, \emptyset, \R, k,l)$ returns a solution for \mmcu\ instance. To bound the size of the solutions returned for an instance of \bmmcu\ we observe the following. 

\begin{align*}
|\mathbb{P}(\I_b)| &\leq (k+1)(l+1)(1+|T_b|(|T_b|+(k+l)(k+l+1)))^{|T_b|} \\
 &\leq (k+1)(l+1)(1+2(k+l)^2(k+l+3))^{2(k+l)}  \\
 &= 2^{\Oh((k+l) \log (k+l))} 
\end{align*}

This is true because $\R_b$ has at most $(k+l)(k+l+1) +|T_b|$ equivalence classes, $E_b$ has at most $T_b$ equivalence classes, each $v \in T_b$ can either go to $X_b$ or choose an equivalence class in $\R_b$ and $E_b$, and $k'$ and $l'$ have $k+1$ and $l+1$ possible values respectively. Let $q= (k+2l)(k+1)(l+1)(1+2(k+l)^2(k+l+3))^{2(k+l)}+k+l$, then all output solutions to a \bmmcu\ instance $\I_b$ affect at most $q-(k+l)$ vertices in total. Now we are ready to prove the lemma which is central for the recursive understanding step. 

\begin{lemma}\label{lemma:recursive}
Assume we are given a \bmmcu~instance $\I_b = (G, T, \R, k,l, T_b)$ and two disjoint sets of vertices $Z,V^* \subseteq V (G)$, such that $|Z| \leq k+l$, $Z \cap T = \emptyset$, $Z_W := N_G(V^*) \subseteq Z$, $|V^* \cap T_b| \leq  k+l$ and the subgraph of $G$ induced by $W := V^* \cup Z_W$ is connected. Denote $G^* = G[W ]$, $T_b^* = (T_b \cup Z_W ) \cap W$, $T^* = T \cap W$, $\R^* = \R|_{T \cap W}$ and $\I^* = (G^* , T^* , R^*, k, l, T_b^*)$. Then $\I^*$ is a proper \bmmcu~instance. Moreover, if we denote by $(\sol_{\PP^*}^*)_{\PP^* \in \mathbb{P}(\I_b^*)}$ an arbitrary output to the \bmmcu~instance $\I_b^*$ and
$$U(\I_b^*) = T_b^* \cup  \{v \in V(G)\ | \ \PP^* \in  \mathbb{P}(\I_b^*), \sol_{\PP^*}^*= \X_{\PP^*}^* \neq \bot   \text{ and } \X_{\PP^*}^* \text{ affects v}  \},$$
then there exists a correct output $(\sol_\PP)_{\PP \in \mathbb{P}(\I_b)}$ to the \bmmcu~instance $\I_b$ such that whenever $\sol_\PP = \X_\PP \neq \bot$ and $\X_\PP$ is a minimal solution to $(\I_b,\PP)$ then $V(\X_\PP) \cap V^* \subseteq U(\I_b^*)$.
\end{lemma}

\begin{proof}
To see that $\I_b^*$ is a proper \bmmcu\ instance, we observe that $G^* = G[W]$ is connected, $|Z_W| \leq |Z| \leq k+l$ and $|V^* \cap T_b| \leq k+l$, and hence $|T_b^*| \leq 2(k+l)$.

To show the second part of the lemma, let us assume $\PP = (X_b, E_b, \R_b, k', l') \in \mathbb{P}(\I_b)$. Let $\X_\PP = (X,F)$ be a solution to $(\I_b, \PP)$. We show that there exists another solution $\X_\PP'=(X',F')$ to  $(\I_b, \PP)$ such that $|X'| \leq |X|$, $|F'| \leq |F|$ and $V(\X_\PP') \cap V^* \subseteq U(\I_b^*)$. To that end, we define a tuple $\PP^* = (X_b^*, E_b^*, \R_b^*, k^*,l^*)$ as follows. Let $F'' := F \setminus E(G[W])$, $X'' := (X \setminus V^*)$ and $\X_\PP'' =  (X'', F'')$. Let $G_\PP'$ be the graph such that $V(G_\PP') = V(G) \setminus W$ and $E(G_\PP') = E(G_\PP - V^*) \setminus E(G_\PP[Z_W]$. Observe that $V(G^*) \cap V(G_\PP') = Z_W$ and $E(G^*) \cap E(G_\PP') = \emptyset$.
\begin{itemize}
 \item $X_b^* = X \cap T_b^*$. 
 \item $E_b^*$ is an equivalence relation on $T_b^* \setminus X_b^*$ such that for $(u,v) \in E_b^*$ if and only if $u$ and $v$ are in the same connected component of $G_\PP' - \X_\PP''$. 
 \item $\R_b^*$ is an equivalence relation on $T^* \cup (T_b^* \setminus X_b^*)$ such that $(u,v) \in \R_b^*$ if and only if $u$ and $v$ are in the same connected component of $G_\PP - \X_\PP$. 
 \item $k^* := |X \cap W|$ and $l^*:= |F \cap E(G[W])|$.
\end{itemize}

Since $|X| \leq k' \leq k$ and $|F| \leq l' \leq l$, we have that $k^* \leq k$ and $l^* \leq l$.
 Also, $E_b^* \subseteq \R_b^*$ because they are both corresponding to relation of being in the same connected component, but the graph considered for defining $E_b^*$ is a subgraph of the one considered for defining $\R_b^*$. Hence we have that $\PP^* \in \mathbb{P}(\I_b^*)$. 
We claim that $(X \cap W, F \cap E[(G[W]))$ is a solution to $(\I_b^*, \PP^*)$. By definition, we have that $|X \cap W| \leq k^*$ and $|F \cap E[(G[W])| \leq l^*$. Also, by definition of $X_b^*$, we have that $X\cap W \cap T_b^* = X_b^*$. 
Now look at two vertices $u, v \in T^* \cup (T_b^* \setminus X_b^*)$. By definition, we have that  $ (u,v) \in \R_b^*$ if and only if there exists a path $P$ connecting $u$ and $v$ in $G_\PP - \X_\PP$. We claim that such a path $P$ exists if and only if there exists a path $P^*$ connecting $u$ and $v$ in $G_{\PP^*}^* - (X \cap W, F \cap E(G[W]))$. It is indeed so, because  subpath of $P$ with internal vertices in $V(G) \setminus W$ corresponds to an edge in $E_b^*$ and vice versa. Thus, $u$ and $v$ are in the same connected component of $G_{\PP^*}^*  - (X \cap W, F \cap E[(G[W]))$ if and only if $(u, v) \in R_b^*$ and hence $(X \cap W, F \cap E[(G[W]))$ is a solution to $(\I_b^*, \PP^*)$.

So we get that $\sol_{\PP^*}^* = \X_{\PP^*}^* =: (X^*, F^*) \neq \bot$ and $|X^*| \leq k^*$ and $|F^*| \leq l^*$. Let $\X_\PP' = (X', F')$, where $X' = (X \setminus W) \cup X^*$ and $F' = (F \setminus E(G[W])) \cup F^*$. We claim that $\X_\PP'$ is a solution for $(\I_b, \PP)$. As $|X^*| \leq |X \cap W|$ and $|F^*| \leq |F \cap E(G[W])$, we have that $|X'| \leq |X| \leq k'$ and $|F'| \leq |F| \leq l'$. Also, since $X_b^* = X \cap T_b^*$ by definition of  $X_b^*$ and by properties of $\X_{\PP^*}^*$, we have that $X_b^* = X^* \cap T_b^*$. Hence,  $X \cap T_b^* = X^* \cap T_b^*$ and $(X,F)$ is a solution for $(\I_b,\PP)$ which gives $X' \cap T_b = X_b$. 

Now let us look at two vertices $u,v  \in T \cup (T_b \setminus X_b)$. We need to show that $u$ and $v$ lie in the same connected component of $G_\PP - \X_\PP'$ if and only if they lie in the same connected component of $G_\PP - \X_\PP$. This will conclude the proof of the lemma. 

Let $u,v  \in T \cup (T_b \setminus X_b)$ lie in the same connected component of $G_\PP - \X_\PP$. Let $D := T \cup (T_b \setminus X_b) \cup Z_W$. Let $P$ be a path connecting $u$ and $v$ in $G_\PP - \X_\PP$ and let $u=v_0, v_1, v_2, \ldots, v_r=v$ be the sequence of vertices that lie on $P$ in that order such that $v_i \in D$ for all $i \in \{0,1,\ldots, r\}$. Since $X \setminus V^* = X' \setminus V^*$ and $X \cap D = X' \cap D = X_b \cup X_b^*$, we have $v_i \notin X'$ for all $i \in \{0,1,2,\ldots,r\}$. Now, to finish this direction of the proof, we need to show that the vertices $v_i$ and $v_{i+1}$ lie in same connected component of $G_\PP - \X_\PP'$ for all $i \in \{0,1, \ldots, r-1\}$.

Let $P_i$ be the subpath of $P$ between $v_i$ and $v_{i+1}$. As $Z_W \subseteq D$, $P_i$ is either a path in $G_\PP'$ or a path in $G_\PP[W]$. For the first case, we know that $X \setminus V^* = X' \setminus V^*$ and also we have that $F \setminus F^* = F' \setminus F^*$, and hence the path is also present on $G_\PP'-\X_\PP$ and hence in $G_P - \X_\PP'$. In the second case, we have that $(v_i, v_{i+1}) \in \R_b^*$. As $(X' \cap W, F' \cap E(G[W])) = (X^*, F^*)$ is a solution to $(\I_b^*, \PP^*)$, we infer that $v_i$ and $v_{i+1}$ are connected via a path $P_i^*$ in $G_{\PP^*}^* - (X' \cap W, F' \cap E(G[W]))$. Hence, $u$ and $v$ are connected in $G_\PP - 
X_\PP'$ if the path $P^*$ does not use any edge $w_1w_2 \in E_b^*$. By definition of $E_b^*$, $w_1w_2 \in E_b^*$ if and only if $w_1$ and $w_2$ lie in the same connected component of $G_\PP'- (X'', F'')$. But we know that $X'' = X \setminus V^* = X' \setminus V^*$ and $F'' = F \setminus E(G[W]) = F' \setminus E(G[W])$ and hence $w_1$ and $w_2$ are also connected in $G_\PP' - \X_\PP''$ and hence in $G_\PP - \X_\PP'$. This completes one direction of the proof. 

For the other direction, the proof is completely symmetrical and we omit the details. This completes the proof of the lemma.   
\end{proof}

Now we describe the recursive step of the algorithm. 

\begin{step}\label{step:recursive}
Assume we are given a \bmmcu\ instance $\I_b = (G, T, \R, k, l, T_b)$. Invoke first the algorithm of Lemma~\ref{lemma:nodesep} in a search for $(q, k+l)$-good node separation (with $V^\infty = T$). If it returns a good node separation $(Z, V_1 , V_2)$, let $j \in \{1, 2\}$ be such that $|V_j \cap T_b| \leq k+l$ and
denote $Z^* = Z$, $V^* = V_j$. Otherwise, if it returns that no such good node separation exists in
$G$, invoke the algorithm of Lemma~\ref{lemma:flowersep} in a search for $(q, k+l)$-flower separation w.r.t. $T_b$ (with $V^\infty = T$ again). If it returns that no such flower separation exists in $G$, pass the instance $\I_b$
to the next step. Otherwise, if it returns a flower separation $(Z, (V_i)_{i=1}^l)$, denote $Z^* = Z$ and
$V^* = \bigcup_{i=1}^l V_i$.

In the case we have obtained $Z^*$ and $V^*$ (either from Lemma~\ref{lemma:nodesep} or Lemma~\ref{lemma:flowersep}), invoke the algorithm recursively for the \bmmcu\ instance $\I_b^*$ defined as in the statement of Lemma~\ref{lemma:recursive} for separator $Z^*$ and set $V^*$, obtaining an output $(\sol_{\PP^*}^*){\PP^* \in \mathbb{P}(\I_b^*)}$. Compute the set $U (\I_b^*)$. Bypass (in an arbitrary order) all vertices of $V^* \setminus (T \cup U(\I_b^*))$. Recall that $T_b^* \subseteq U(\I_b^*)$, so no border terminal gets bypassed. After all vertices of $V^* \setminus U(\I_b^*)$ are bypassed, perform the following operations on terminals of $V^* \cap T$:
\begin{enumerate}
 \item  As long as there exist two different $u, v \in V^* \cap T$ such that $(u, v) \notin \R$, and $uv \in E(G)$, then delete the edge $uv$ and decrease $l$ by $1$; if $l$ becomes negative by this operation, return $\bot$ for all $\PP \in  \mathbb{P}(\I_b)$. 
 \item  As long as there exist two different $u,v \in V^* \cap T$ such that $(u, v) \in \R$ and  either $uv \in E(G)$ or $|N_G(u)\cap N_G(v)| > k+l$, identify $u$ and $v$.
 \item  If the above two rules are not applicable, then, as long as there exist three pairwise distinct terminals
$u_1, u_2, \ldots, u_t \in T$ of the same equivalence class of $\R$ that have the same neighborhood and $t > l+2$, delete $u_i$ for $i > l+2$ from the graph (and delete all pairs containing $u_i$ from $\R$).
\end{enumerate}
 Let $\I_b'$ be the outcome instance.

Finally, restart this step on the new instance $\I_b'$ and obtain a family of solutions $(\sol_\PP)_{\PP \in \mathbb{P}(\I_b')}$ and return this family as an output to the instance $\I_b$.
\end{step}

We first verify that the application of Lemma~\ref{lemma:recursive} is justified in Step~\ref{step:recursive}. By definitions of good node separations and flower separations and by choice of $V^*$, we have that $|V^* \cap T_b| \leq k+l$ and $G[V^* \cup N_G(V^*)]$ is connected. Also, the recursive calls are applied to strictly smaller graphs because in good node separation, $V_2$ is deleted in the recursive call while in the other case, by definition of flower separation, we have that $Z \cup \bigcup_{i=1}^l V_i$ is a proper subset of $V(G)$.

After the bypassing operations, we have that $V^*$ contains at most $q$ vertices that are not terminals (at most $k+l$ border terminals and at most $q - (k+l)$ vertices which are neither terminals nor border terminals). Let us now bound the number of terminal vertices once Step~\ref{step:recursive} is applied. Note that, after Step~\ref{step:recursive} is applied, for any $v \in T \cap V^*$, we have $N_G(v) \subseteq (V^* \setminus T) \cup Z$ and $|(V^* \setminus T) \cup Z| \leq (q + k+l)$. Due to the first and second rule in Step~\ref{step:recursive}, for any set $A \subseteq (V^* \setminus T) \cup Z$ of size $k + l+ 1$, at most one terminal of $T \cap V^*$ is adjacent to all vertices of $A$. Due to the third rule in Step~\ref{step:recursive}, for any set $B \subseteq (V^* \setminus T) \cup Z$ of size at most $k+l$ and for each equivalence class of $\R$, there are at most $l+2$ terminals of this equivalence class with neighborhood exactly $B$. Let $q' := |T \cup V^*|$, then we have the following.

$$ q' \leq (q+k+l)^{k+l+1} + (l+2)(k+l)(k+l+1) \sum_{i=1}^{k+l}(q+k+l)^i = 2^{\Oh((k+l)^2 \log (k+l))} $$

\begin{lemma}
 Assume that we are given a \bmmcu\ instance $\I_b = (G, T, \R, k, l, T_b)$ on which Step~\ref{step:recursive} is applied, and let $\I_b'$ be an instance after Step~\ref{step:recursive} is applied. Then any correct output to the instance $\I_b'$ is a correct output to the instance $\I_b$ as well. Moreover, if Step~\ref{step:recursive} outputs $\bot$ for all $\PP \in  \mathbb{P}(\I_b')$, then this is a correct output to $\I_b$.
\end{lemma}

\begin{proof}
 We first note that by Lemma~\ref{lemma:recursive}, for all $\PP \in \mathbb{P}(\I_b)$, for all the vertices $v \notin U(\I_b^*)$, there exists a minimal solution to $(\I_b, \PP)$ that does not affect $v$, hence by Lemma~\ref{lemma:bypass}, the bypassing operation is justified. The second and third rules are justified by Lemma~\ref{lemma:identify} and Lemma~\ref{lemma:deleteterminals} respectively. The first rule is justified by Lemma~\ref{lemma:certainedge}, and if application of this rule makes $l$ negative then  for any $\PP \in \mathbb{P}(\I_b)$, there is no solution to $(\I_b, \PP)$.
\end{proof}

Now we do a running time analysis for Step~\ref{step:recursive}. The applications of Lemma~\ref{lemma:nodesep} and Lemma~\ref{lemma:flowersep} take time $\Oh (2^{\Oh(\min(q,k+l) \log(q+k+l))} n^3 \log n) = \Oh (2^{\Oh((k+l)^2 \log(k+l))} n^3 \log n)$. Let $n' = |V^*|$; the recursive step is applied to a graph with at most $n' + k+l$ vertices and, after bypassing, there are at most $\min(n - 1, n - n' + q + q')$ vertices left. Moreover, each bypassing operation takes $\Oh(n^2)$ time, the computation of $U(\I_b^*)$ takes $\Oh(2^{\Oh((k+l) \log (k+l))} n)$ time. Each application of Lemma~\ref{lemma:certainedge} takes $\Oh(n^2)$ time and it can be applied at most $l$ times. So all applications of Lemma~\ref{lemma:certainedge} take $\Oh(ln^2)$ time. Application of Lemma~\ref{lemma:identify} takes $\Oh((k+l)n^2)$ time per operation, which can be implemented by having a counter for each pair of terminals and increasing those counters accordingly by considering every pair of terminals of $N_G(x)$, for each $x \in V(
G) 
\setminus T$. Since when a counter reaches value $k+l + 1$ for vertices $u, v$, we know that $|N_G(u)\cap N_G(v)| > k+l$, the total time consumed is bounded by $\Oh((k+l)n^2)$. Application of Lemma~\ref{lemma:deleteterminals} takes $\Oh(n^2 \log n)$ time per operation, since we can sort terminals from one equivalence class according to their sets of neighbours. Since both Lemma~\ref{lemma:identify} and Lemma~\ref{lemma:deleteterminals} decrease the number of terminals by at least $1$, when applied, they can be applied at most $n$ times.  Hence, all applications of Lemma~\ref{lemma:identify} and Lemma~\ref{lemma:deleteterminals} takes total $\Oh(n^3 (k+l + \log n))$ time. We also note that values of $k$ and $l$ do not change during the recursive calls. So, we get the following recurrence relation for the running time of Step~\ref{step:recursive} as a function of vertices of the graph $G$. 

$$T(n) \leq \max_{q+1 \leq n' \leq n-q-1} \big( \Oh (2^{\Oh((k+l)^2 \log(k+l))} n^3 \log n) + T(n'+k+l) + T(\min(n-1,n-n'+q+q')) \big) $$

Solving the recurrence gives $T(n) = \Oh (2^{\Oh((k+l)^2 \log(k+l))} n^4 \log n)$ in the worst case, which is the desired upper bound for the running time of the algorithm.

%
%
%
%
%

\subsection{High Connectivity phase}

In this section we describe the high connectivity phase for the algorithm. Assume we have a \bmmcu\ instance $\I_b = (G, T, \R, k', l', \V,  T_b)$ where Step~\ref{step:recursive} is not applicable. Let us fix $\PP = (X_b, E_b, \R_b, k,l) \in \mathbb{P}(\I_b)$. We iterate through all possible values of $\PP$ and try to find a minimal solution to $(\I_b, \PP)$. Since $|\mathbb{P}(\I_b)| = 2^{\Oh((k+l) \log (k+l))}$ it results in a factor of $2^{\Oh((k+l) \log (k+l))}$ in the running time.  For a graph $G$, by $\CS(G)$ we denote the set $V(G) \cup E(G)$. Similarly, for a tuple $\X = (X, F)$, by $\CS(\X)$ we denote the set $X \cup F$. Since we have assumed that Step~\ref{step:recursive} is not applicable, for any  $\X = (X,F)$ where $X \subseteq V(G) \setminus T$, $F \subseteq E(G)$, $|X| \leq k$ and $|F| \leq l$, Lemma~\ref{lemma:highconnected} implies that the graph $G - \X$ contains at most $t := (2q + 2)(2(k+l) - 1) + 2(k+l) + 1$ connected components containing a non-terminal, out of which at most 
one can 
contain 
more than $q$ vertices outside $T$. Let us denote its vertex set by $\bigs(\X)$ (observe that this can possibly be the empty set, in case such a component does not exist). We once again need to use lemmas~\ref{lemma:certainedge}-\ref{lemma:deleteterminals} to bound number of terminals. 

\begin{step}\label{step:terminalreduction}
 Apply Lemma~\ref{lemma:certainedge}, Lemma~\ref{lemma:identify} and Lemma~\ref{lemma:deleteterminals} exhaustively on the set $T$ of terminals in the graph (as done in rules 1-3 of Step~\ref{step:recursive}, but doing it for all of $T$ instead of just $T \cap V^*$). Then remove all but $l+1$ copies of edges having multiplicity more than $l+1$.
\end{step}

The running time analysis of applying lemmas~\ref{lemma:certainedge}-\ref{lemma:deleteterminals} in this step is exactly the same as the one done in Step~\ref{step:recursive} and takes $\Oh(n^3(k+l+\log n))$ time. Deletion of copies of parallel edges is justified by Lemma~\ref{lemma:deleteparallel} and takes time at most $\Oh(mn)$. Hence, the total time taken by this step is $\Oh(n^3(k+l+\log n))$. Now we define the notion of interrogating a solution, which will help us in highlighting the solution. 

\begin{definition}
Let  $\Z = (Z,F')$ where $Z \subseteq V(G) \setminus T$, $F' \subseteq E(G)$, $|Z| \leq k$ and $|F'| \leq l$ and let $S \subseteq \CS(G) \setminus T$. We say that $S$ interrogates $\Z$ if the following holds:
\begin{itemize}
 \item $S \cap \CS(\Z) = \emptyset$;
 \item for any connected component $C$ of $G - \Z$ with at most $q$ vertices outside $T$, all vertices and edges of $C$ belong to $S \cup T$.
\end{itemize}
\end{definition}

\begin{lemma}\label{lemma:setfamilyinterrogates}
Let $q''= (qt+k+l)^{k+l+1} + (l+2)(k+l)(k+l+1) \sum_{i=1}^{k+l} (qt+k+l)^i$.  Let $\F$ be a family obtained by the algorithm of Lemma~\ref{lemma:highconnected} for universe $U = \CS(G) \setminus T$ and constants $a = q''+ qt +{q''+qt \choose 2}$ and $b = k+l$, Then, for any $\Z = (Z,F')$ where $Z \subseteq V(G) \setminus T$, $F' \subseteq E(G)$, $|Z| \leq k$ and $|F'| \leq l$, there exists a set $S \in \F$ that interrogates $\Z$.
\end{lemma}

\begin{proof}
 Let $\Z = (Z,F')$ where $Z \subseteq V(G) \setminus T$, $F' \subseteq E(G)$, $|Z| \leq k$ and $|F'| \leq l$. Let $A$ be the union of vertex sets of all connected components of $G - \Z$ that have at most $q$ vertices outside $T$ and let $B$ be the set of edges of these connected components. By Lemma~\ref{lemma:highconnected}, $|A \setminus T | \leq qt$. Also, since we have applied Step~\ref{step:terminalreduction} exhaustively, we have that no two terminals are adjacent in the graph. Let $C = \{u \ | uv \in E(G) \cap F', v \in A\cap T\}$. Clearly, $C \cap T = \emptyset$, $|C| \leq l$ and for any $v \in A \cap T$, $N(v) \subseteq (A  \setminus T) \cup Z \cup C$ which gives $|N(A \cap T)| \leq |(A  \setminus T) \cup Z \cup C| \leq qt + k+l$. Now doing the same analysis as in Step~\ref{step:recursive} gives us $|A \cap T| \leq q''$. Since all the edges in the graph have multiplicity at most $l+1)$, we have that $|A \cup B| \leq q''+ qt+ (l+1){q''+qt \choose 2}$ and $|\CS(\Z)| \leq k+l$. So, by Lemma~\ref{lemma:
setfamily}, there is a set $S \in \F$ that is 
disjoint with $\CS(\Z)$ 
and contains $(A \setminus T) \cup B$. By construction of $A$ and $B$, $S$ interrogates $\Z$. This completes the proof of the lemma.
\end{proof}

\begin{step}\label{step:setfamily}
Compute the family $\F$ from Lemma~\ref{lemma:setfamilyinterrogates} and branch into $|\F|$ subcases, indexed by sets $S \in \F$. In a branch $S$ we seek for a minimal solution $\X_\PP$ to $(\I_b , \PP)$, which is interrogated by $S$.
\end{step}

Note that since we have $q'' = 2^{\Oh((k+l)^2 \log (k+l))}$ and $q, t = 2^{\Oh((k+l) \log (k+l))}$, the family $\F$ of Lemma~\ref{lemma:setfamily} is of size $\Oh (2^{\Oh((k+l)^3 \log(k+l))} \log n)$ and can be computed in $\Oh (2^{\Oh((k+l)^3 \log(k+l))} n \log n)$ time. The correctness of Step~\ref{step:setfamily} is obvious from Lemma~\ref{lemma:setfamilyinterrogates}. As discussed, it can be applied in $\Oh (2^{\Oh((k+l)^3 \log(k+l))}n \log n)$ time and gives rise to $\Oh (2^{\Oh((k+l)^3 \log(k+l))} \log n)$ subcases. 

\begin{lemma}\label{lemma:highconnectivity}
Let $\X_\PP = (X,F)$ be a solution to $(\I_b, \PP)$ interrogated by $S$. Then there exists a set $T^{\bigs} \subseteq T \cup (T_b \setminus X_b)$ that is empty or contains all vertices of exactly
one equivalence class of $\R_b$, such that $X \subseteq (X_b \cup N_G (S(T^{\bigs}))$ and $F = A_{G,X}(S(T^\bigs))$, where $S(T^{\bigs})$ is the union of vertex sets of all connected components of $G(S \cup T \cup (T_b \setminus X_b))$ that contain a vertex of $(T \cup (T_b \setminus X_b )) \setminus T^{\bigs}$ and $A_{G,X}(S(T^\bigs))$ is set of edges in $G$ which have at least one end point in $S(T^\bigs)$ but do not belong to any of the connected components of $G[S(T^\bigs)]$ and are not incident on $X$.

\end{lemma}

\begin{proof}
Let $\X_\PP = (X,F)$. Consider the graph $G_\PP-\X_\PP$ and let $\bigs_\PP(\X_\PP)$ be the vertex set of the connected component of $G_\PP -\X_\PP$ that contains $\bigs(\X_\PP)$ (recall that $G_\PP$ is the graph $G$ with additional edges $E_b$ ; thus $\bigs_\PP(\X_\PP)$ may be significantly larger than $\bigs(\X_\PP)$). If $\bigs(\X_\PP) = \emptyset$, then we take $\bigs_\PP(\X_\PP)$ to be the empty set as well. As $\X_\PP$ is a solution to $(I_b, \PP)$, we have $X \cap T_b = X_b$. Define $T^\bigs = (T \cup (T_b \setminus X_b)) \cap \bigs_\PP (\X_\PP )$; as $\X_\PP$ is a solution to $(I_b , \PP)$, $T^\bigs$ is empty or contains vertices of exactly one equivalence class of $R_b$.

Now let $C$ be the vertex set of a connected component of $G-\X_\PP$ that contains a vertex $v \in (T \cup (T_b \setminus X_b)) \setminus T^\bigs$. Clearly, $v \notin \bigs_\PP(\X_\PP)$. As $S$ interrogates $\X_\PP$, $\bigs_\PP (\X_\PP)$ contains $\bigs(\X_\PP)$ and $X \cap (T \cup T_b) = X_b \subseteq T_b$, we infer that $C$ is the vertex set of a connected component of $G(S \cup T \cup (T_b \setminus X_b ))$ as well. As $v \in C$, $C$ is a connected component of $G[S(T^\bigs)]$. Let $X' = X_b \cup (X \cap N_G(S(T^\bigs)))$ and let $F'= A_{G,X'}(S(T^\bigs)) \cap F$. We claim that $\X'=(X',F')$ is a solution to $(\I_b, \PP)$. 

As $X_b \subseteq X$, we have that $X' \subseteq X$ and hence $|X'| \leq k$. Also, since $F' \subseteq F$, we have that $|F'| \leq l$. Moreover, as $X' \subseteq X$, $F' \subseteq F$ and $\X_\PP = (X,F)$ is a solution to $(\I_b , \PP)$, if $(u, v) \in \R_b$ then $u$ and $v$ are in the same connected component of $G_\PP - \X'_\PP$. We now show that for any $(u, v) \notin \R_b$ the vertices $u$ and $v$ are in different connected components of $G_\PP - \X'_\PP$. For the sake of contradiction, let $u,v \in T \cup (T_b \setminus X_b)$ such that $(u,v) \notin \R_b$, $u$ and $v$ lie in the same connected component of $G_\PP - \X'_\PP$ and distance between $u$ and $v$ is minimum possible. Let $P$ be the shortest path between $u$ and $v$ in $G_\PP - \X'_\PP$.

As $\X_\PP = (X,F)$ is a solution to $(\I_b , \PP)$, $u$ and $v$ lie in different connected components of $G_\PP - \X_\PP$. Without loss of generality, let $v \notin T^\bigs$ and let $C$ be the connected component of $G - \X_\PP$ that contains $v$. Since $(u,v) \notin \R_b$ and $G-\X_\PP$ is a subgraph of $G_\PP - \X_\PP$, we have that $u \notin C$. Since for every edge $xy$ in $G$ such that $x \in C$ and $y \notin C$, we have that either $y \in X'$ or $xy \in F'$, the path $P$ contains an edge $v_1u_1 \in E_b$ such that $v_1 \in C$ and $u_1 \notin C$. So in that case we have that $v_1, u_1 \in T_b$ and hence $(v_1,u_1) \in \R_b$. Also, since $v_1 \in C$, $(v,v_1) \in \R_b$ and since $\R_b$ is an equivalence relation, we have that $(v,u_1) \in \R_b$. But by our assumption, $(u,v) \notin \R_b$, which gives that $(u_1, u) \notin \R_b$. Hence $u_1, u \in T \cup (T_b \setminus X_b)$ such that $(u_1, u) \notin \R_b$  and they are connected by a proper subpath of $P$ in $G_\PP - \X'_\PP$, which contradicts our 
choice of $u$, $v$ and $P$. This completes the proof of the lemma.  
\end{proof}

Now we are ready to give the final step of the algorithm. The correctness of the step follows from Lemma~\ref{lemma:highconnectivity} and the fact that if $S$ interrogates a solution $\X$ to $(\I_b, \PP)$, then $|N_G(S(T^\bigs))| \leq k+l$. 

\begin{step}\label{step:finalstep}
For each branch, where $S$ is the corresponding guess, we do the following. For each set $T^\bigs$ that is empty or
contains all vertices of one equivalence class of $\R_b$, if $|N_G(S(T^\bigs))| \leq k+l$, then for each  $X  \subseteq X_b \cup N_G(S(T^\bigs))$ such that $|X| \leq k$, and  $F = A_{G,X}(S(T^\bigs))$, check whether $(X,F)$ is a solution to $(\I_b,\PP)$ interrogated by $S$. For each $\PP$, output a minimal solution to $(\I_b,\PP)$ that is interrogated by $S$. Output $\bot$ if no solution is found for any choice of $S$, $T^\bigs$ and $X$. 
\end{step}

Note that $\R$ has at most  $(k+l)(k+l+1)$ equivalence classes. As $|T_b| \leq 2(k+l)$, we have $\R_b \leq (k+l)(k+l+3)$, and hence there are at most $(k+l)(k+l+3)+1$ choices of the set $T^\bigs$. For each $T^\bigs$, computing $N_G(S(T^\bigs))$ and checking whether $|N_G(S(T^\bigs))| \leq k+l$ takes $\Oh(n^2)$ time. Since $X_b \leq k$, there are at most $(k+1)(2k+l)^k$ choices for $X$, and then computing $F = A_{G,X}(S(T^\bigs))$ and checking whether $(X,F)$ is a solution to $(\I_b,\PP)$ interrogated by $S$ take $\Oh(n^2)$ time each. Finally, checking whether the solution is minimal or not and computing a minimal solution takes additional $\Oh((k+l)n^2)$ time. Therefore Step~\ref{step:finalstep} takes $\Oh (2^{\Oh((k+l)^3 \log(k+l))} n^2 \log n)$ time for all subcases.

This finishes the description of fixed-parameter algorithm for \mmcu\ and we get the following theorem.

\begin{theorem}
 \mmcu\ can be solved in $\Oh (2^{\Oh((k+l)^3 \log(k+l))} n^4 \log n)$ time.
\end{theorem}






\bibliographystyle{abbrv}
\bibliography{references}

\end{document}